\documentclass[a4paper,UKenglish,cleveref,autoref,thm-restate]{lipics-v2021}

\pdfoutput=1

\hideLIPIcs

\nolinenumbers

\let\autocite\cite

\usepackage{syntax}
\usepackage{proof}
\usepackage{units}
\usepackage{multirow}
\usepackage{amsthm}
\usepackage{amssymb}
\usepackage{stmaryrd}
\usepackage{eucal}
\DeclareUnicodeCharacter{2200}{\forall}
\DeclareUnicodeCharacter{2208}{\in}
\DeclareUnicodeCharacter{2264}{\leq}
\DeclareUnicodeCharacter{2219}{\cdot}
\DeclareUnicodeCharacter{03B1}{\alpha}
\DeclareUnicodeCharacter{03B2}{\beta}
\DeclareUnicodeCharacter{03BB}{\lambda}
\DeclareUnicodeCharacter{0393}{\Gamma}
\DeclareUnicodeCharacter{03A3}{\Sigma}
\DeclareUnicodeCharacter{2203}{\exists}
\DeclareUnicodeCharacter{2227}{\land}
\newcommand{\sequent}[4]{\ensuremath{#1 \vdash^{#2}_{#3} #4}}
\newcommand{\Type}[0]{\ensuremath{\mathbf{Type}}}
\newcommand{\isType}[1]{\ensuremath{#1 \; \mathbf{type}}}
\newcommand{\bigoh}[1]{\ensuremath{\mathcal{O}\left(#1\right)}}

\newcommand{\N}[0]{\ensuremath{\mathbb{N}}}
\newcommand{\R}[0]{\ensuremath{\mathbb{R}}}
\newcommand{\rulesep}[0]{\quad \quad \quad}
\newcommand{\Infer}[3][]{\ensuremath{\infer[\textit{#1}]{#3}{#2}}}
\newcommand{\interp}[1]{\ensuremath{\llbracket #1 \rrbracket}}
\newcommand{\splice}[1]{\ensuremath{\big( #1 \big) }}
\newcommand{\mathOp}[3]{\ensuremath{\mathop{\textit{#1}_{#2}#3}}}
\newcommand{\iter}[4]{\mathOp{iter}{}{\left(\lambda{}{#1}.{#2}, {#3}, {#4}\right)}}
\newcommand{\subst}[2]{\ensuremath{\llbracket \nicefrac{#2}{#1} \rrbracket}}
\newcommand{\occurs}[2]{\mathOp{\textit{occurs}}{}{\left({#1},{#2}\right)}}
\newcommand{\inr}[1]{\mathOp{in}{r}{\left(#1\right)}}
\newcommand{\inl}[1]{\mathOp{in}{l}{\left(#1\right)}}
\newcommand{\dbi}[1]{\mathop{\mathbb{B}}{\left(#1\right)}}
\newcommand{\prl}[1]{\mathOp{prj}{l}{#1}}
\newcommand{\prr}[1]{\mathOp{prj}{r}{#1}}
\newcommand{\pair}[2]{\ensuremath{\left(#1,#2\right)}}
\newcommand{\depth}[1]{\ensuremath{|#1|_d}}

\newcommand{\depprod}[5]{\ensuremath{\Pi(#1:#2)\rightarrow^{#3}_{#4}{#5}}}

\newcommand{\witheq}[1]{\overline{\overline{#1}}}

\AtBeginDocument{\renewcommand{\setminus}{\mathbin{\backslash}}}
\renewcommand{\max}[2]{\ensuremath{\mathop{max}(#1,#2)}}
\newcommand{\eval}[3]{#2 \big\Downarrow_{#1} #3}
\newcommand{\rec}[3]{\mathop{\textit{rec}}(#1, #2, #3)}
\newcommand{\lastpage}[1]{}
\DeclareMathOperator{\unitType}{\circ}
\DeclareMathOperator{\unitValue}{∙}
\DeclareMathOperator{\Zero}{zero}
\DeclareMathOperator{\aSucc}{succ}
\DeclareMathOperator{\Nat}{Nat}

\DeclareMathOperator{\Bound}{Bound}
\DeclareMathOperator{\Var}{Var}

\newcommand{\bOne}[0]{\mathbf{1}}
\newcommand{\bSucc}[1]{\mathbf{S}(#1)}
\newcommand{\hyper}[3]{\mathop{hyper}(#1,#2,#3)}
\newcommand{\ack}[2]{\mathop{ack}(#1,#2)}
\newenvironment{left-cases}[0]{\left\lbrace\begin{aligned}}{\end{aligned}\right.}
\newenvironment{caseof}[1]{\mathop{\textit{case}} \; #1 \mathop{\textit{of}} \; \begin{left-cases}}{\end{left-cases}}

\newcommand{\case}[2]{#1 & \rightarrow & #2; \\}
\DeclareMathOperator{\alt}{\big|}
\DeclareMathOperator{\free}{free}

\title{Consistent Ultrafinitist Logic}

\titlerunning{Ultrafinitist Logic}

\author{Michał J. {Gajda}}{Migamake Pte
Ltd,  \url{https://www.migamake.com} }{mjgajda@migamake.com}{https://orcid.org/0000-0001-7820-3906}{}

\authorrunning{MJ Gajda}

\Copyright{Michał J. Gajda}



\category{} 

\acknowledgements{We thank Vendran Čačić, Seth Chaiken, Bhupinder Singh
Anand, Orestis Melkonian, Anton Setzer for their invaluable feedback.
Thanks to Daniel Schwartz for sparkling my interest in Kleene normal
form.}

\providecommand{\tightlist}{%
  \setlength{\itemsep}{0pt}\setlength{\parskip}{0pt}}

\relatedversion{Presented on TYPES2022, CiE2022, LAP2021}
\relatedversiondetails[linktext={arXiv:2106.13309}]{Preprint}{https://arxiv.org/abs/2106.13309}

\begin{document}

\maketitle

\begin{abstract}
Ultrafinitism postulates that we can only compute on relatively short
objects, and numbers beyond a certain value are not available. This
approach would also forbid many forms of infinitary reasoning and allow
removing certain paradoxes stemming from enumeration theorems. For a
computational application of ultrafinitist logic, we need more than a
proof system, but a logical framework to express both proofs, programs,
and theorems in a single framework. We present its inference rules,
reduction relation, and self-encoding to allow direct proving of the
properties of ultrafinitist logic within itself. We also provide a
justification why it can express all bounded Turing programs, and thus
serve as a ``logic of computability''.
\end{abstract}

Ultrafinitism
\autocite{explicit-finitism,podnieks-finitism,yessenin-criticism,gefter,lenchner}
postulates that we can only reason and compute relatively short
objects\footnote{For example, a computation run by computer the size of
  Earth within the lifespan of Earth so far. Of the order of \(10^{93}\)
  as described by \autocite{bremermann}.}

Tighter limit can be established using petahertz frequency
\autocite{petahertz} as a quantum limit for light-based systems giving
\(10^{33}\) serial cycles during the lifetime of Earth.{]}
\autocite{computational-capacity-of-universe,bremermann,feasible-numbers,ultimate-limits-of-computation,limits-of-computation,small-universe},
and numbers beyond certain value are not available
\autocite{yessenin-criticism,feasible-numbers}. Some philosophers also
question the physical existence of real numbers beyond a certain level
of accuracy \autocite{gisin}. This approach would also forbid many forms
of infinitary reasoning and allow removing many from paradoxes stemming
from a countable enumeration.

However, philosophers\footnote{We cite physical and metaphysical
  arguments from previous work equally.} still disagree on whether such
a ultrafinitist logic could be consistent
\autocite{wangs-paradox,strict-finitism-refuted-q}, while constructivist
mathematicians claim that \emph{``no satisfactory developments exist''}
\autocite{constructivism-review}. We present a proof system based on the
Curry-Howard isomorphism \autocite{curry-howard} and explicit bounds for
computational complexity that answers the question.

This approach invalidates logical paradoxes that stem from a profligate
use of transfinite reasoning
\autocite{benardete,clowns,finitism-vs-pra}, and assures that we only
state problems that are decidable by the limit on input size, proof
size, and the number of steps. This explicitly excludes phenomena of
undecidability by excluding them from our realm of valid statements
\autocite{undecidable-theories}. Our approach allows to express all
Turing Machine programs that are bounded \autocite{time-bounded-turing}
by proof terms of the logic\footnote{Up to fixed emulation overhead, see
  Emulation Complexity below in section \ref{emulation-complexity}.}.

Explicitly bounding computational complexity also prevents a famous
\emph{paradox of inference}. This paradox of classical theory of
semantic information
\autocite{semantic-information,paradox-of-non-trivial-inference}
unjustly labels all mathematical proofs as ``trivial information'',
because it can be inferred from the axioms.

\section{Introduction}\label{introduction}

By \emph{finitism} we understand the mathematical logic that tries to
absolve us from transfinite inductions \autocite{explicit-finitism}.
\emph{Ultrafinitism}\footnote{Also called \emph{strict finitism} by
  \autocite{strict-finitism-refuted-q}.} goes even further by
postulating a definite limit for the complexity of objects that we can
compute with
\autocite{ultimate-limits-of-computation,limits-of-computation,feasible-numbers,computational-capacity-of-universe,bremermann,finite-entropy}.
We assume these without committing ourselves to adopt a fixed number as
a limit.

In order to permit only \emph{ultrafinitist} inferences, we postulate
\emph{ultraconstructivism}: we permit only \emph{constructive} proofs
with a \emph{deadline}. That is constructions that are not just strictly
computable, but for which there is a \emph{upper bound} on the amount of
computation that is needed to resolve them. That means that we forbid
proofs that go for an arbitrarily long time and require \emph{totality}
for any proof or computation.

For the sake of generality, we will attach this \emph{deadline} in the
form of \emph{bounding function} that takes as arguments \emph{size
variables} (\emph{depths of input terms}), and outputs the upper bound
on the number of steps that the proof is permitted to make (along with
upper bound on the size of the output). \emph{Depths of input terms} are
a convenient upper bound on the complexity of normalized proof terms.
(Normalized proof terms are those with opportunity for \emph{cut} or
\(\beta\)-reduction.)

Our approach is inspired by the \emph{Curry-Howard isomorphism} -- the
fact that the constructive proofs always correspond to executable
programs. It also follows \emph{inverse Curry-Howard isomorphism}: the
philosophy that rejects logical inference which do not correspond to
programs computable in our universe\footnote{Given that all formal
  proofs in mathematically strict formal systems can be considered
  finite numbers of connected steps in computation or hypercomputation.}

The philosophy of \emph{ultraconstructivism} would similarly purport
that while transfinitary logics may be consistent, they are correspond
to objects \emph{,,out-of-this-world'\,'}\footnote{Extra-universal.},
since our observable universe is inherently finitary
\autocite{computational-capacity-of-universe}.

\subsection{Contributions}\label{contributions}

To enumerate chief contributions of this paper:

\begin{enumerate}
\def\labelenumi{\arabic{enumi}.}
\tightlist
\item
  First consistent \textbf{ultrafinitist} logic to the knowledge of the
  author. It allows bounding by any and arbitrarily large computational
  limit (section \ref{bounds}), and consistent reasoning resolving
  Wang's Paradox\autocite{wangs-paradox}. Thus this logic is first
  formal theory to claim a purely philosophical legacy of ultrafinitism
  \autocite{yessenin-criticism,feasible-numbers,strict-finitism-refuted-q,varieties-of-finitism}.
\item
  A decidable logic having \textbf{meta-theory expressible in itself}
  (see section \ref{self-encoding}).
\item
  Clear and comprehensive demonstration that \textbf{assumptions of
  Gödel are too strong} \autocite{goedel-undecidability,goedel-works}
  when considering bounded logics. This is because decidability of
  bounded term can be demonstrated by simple enumeration with a more
  generous bound. Most proofs are elementary by enumeration. Proof of
  consistency is done by reduction to widely known intuitionistic logic
  to make the paper accessible to second year students of computer
  science or first year graduates in constructive mathematics (section
  \ref{properties}).
\item
  Candidate for a most expressive logic that allows explicitly bounded
  computable functions. Ideal \textbf{logic of computability} must
  forbid all reasoning about uncomputable, and only allow computable
  statements (with proofs corresponding to \emph{bounded Turing Machine}
  programs.) For all bounds that can be computed within the framework,
  we can also compute the function bounded by these (section
  \ref{turing-completeness}).
\item
  Placing \textbf{ultrafinitism} and \textbf{ultraconstructivism} as
  candidate for realization of \emph{computable foundations for
  mathematics} programme (discussion in section
  \ref{computable-foundations}).
\item
  All statements with bounds having a proof without bounds have a proof
  with bounds too\footnote{Thanks to anonymous reviewer for pointing
    importance of this result.} (see section \ref{finitary-completeness}
  theorem \ref{preservation-of-bounds}).
\end{enumerate}

We will further abbreviate the ``Consistent Ultra-Finitist Logic''
proposed in here as ``UFL'' when speaking about higher order variant
(with dependent \(\Pi\),\(\Sigma\) types for quantification).

\section{Syntax and inference}\label{syntax-and-inference}

Due to size bounds and clarity of this paper, we first introduce
propositional ultrafinitist logic, and then describe interpretation of
universal quantifier in a separate section.

\subsection{Bounds}\label{bounds}

We express bounds as polyvariate functions of the natural numbers,
called \emph{sizes}. These explicitly bound our proofs, depending on the
size of input terms. While subtraction within natural domain is
permitted, only positive results of computation are permitted. All bound
functions are increasing with respect to all arguments (monotonic).

The bounds\footnote{Using a bound on cost and depth of the term for each
  inference, we independently developed a very similar approach to that
  used for cost bounding in higher-order rewriting
  \autocite{higherOrderRewriting}.} will be standing on one of two
roles: as an upper bound on the proof complexity, and there we will use
symbol \(\alpha{}\) as a placeholder, or to state an upper bound on the
depth of the normal form of the proof indicated by the symbol
\(\beta{}\). That is because the number of constructors may sometimes
bound a recursive examination of the proof of a proposition.

Here \(\rho^{\rho}\) is an exponentation, and
\(\iter{v}{\rho_1}{\rho_2}{\rho_3}\) is an iterated composition of
function described by expression \(\rho_1\) with respect to an argument
variable \(v\); that iteration happens \(\rho_2\) times, and the
function is applied to initial argument \(\rho_3\). The
\(\rho_1 \subst{v}{\rho_2}\) describes substitution of bound variable
\(v\) by \(\rho_2\), inside expression \(\rho_1\).

Any total function \(f(...)\) over naturals has a bounds \(b(...)\)
function \(∀ x ∈ \Nat. x ≤ n \implies f(x) ≤ b(n)\): given a range
limited \footnote{This is not true for traditional real numbers \(\R\):
  hyperbola \(y=\frac{1}{x}\) is unbounded around \(0\) because
  \(\lim\limits_{x→0}\frac{1}{x}=-\infty\).} by \(n\), we can compute
\(b(n)\) that is \(\mathop{max}\limits_{x≤n}f(x)\).

\begin{description}
\tightlist
\item[Conjecture Bounds terminate]
Since all iterations in bounds quantification terminate, all bounds
terminate.
\end{description}

\subsection{Terms}\label{terms}

All terms are explicitly limited, but we avoid labelling terms for which
bounds can be easily inferred (see below).

\[ \begin{array}{lrcl}
\text{Size variables:}     & v      & \in & V \\
\text{Size values}         & n      & = & \bOne \alt \bSucc{n} \\
\text{Term variables:}     & x      & \in & X \\
\text{Positive naturals:}  & i      & \in & \N \setminus \lbrace{}0\rbrace{} \\
\text{Upper bounds:}       & \rho   & ::= & v \alt  i \alt \bSucc{\rho} \alt \\
                           &        & \alt & \iter{v}{\rho}{\rho}{\rho} \alt \rho \subst{v}{\rho} \alt \max{\rho}{\rho} \\
\end{array} \]

Iteration is defined as:

\[ \begin{array}{lcl}
\iter{v}{e}{\bOne}{a}               & = & e\subst{v}{a} \\
\iter{v}{e}{\bSucc{n}}{a}           & = & \iter{v}{e}{n}{e\subst{v}{a}} \\
\end{array} \]

Later we will explain how bounds expressions can be encoded in the same
language as the proof terms. At the level of basic logic we do not need
this, but it will become useful when we consider meta-reasoning (and
encode entirety of the logic within its own proof terms.)

\[ \begin{array}{lrcl}
\text{Data size bounds:}   & \alpha & ::= & \rho \\
\text{Computation bounds:} & \beta  & ::= & \rho \\
\text{Types:}              & \tau   & ::= & v \alt \tau \land \tau \alt \tau \lor \tau \alt \tau_v \rightarrow^{\alpha}_{\beta} \tau \alt \bot \alt \unitType \alt \Type \\
\text{Terms:}              & E      & ::= & x \alt \lambda x.  E \alt \inr{E} \alt \inl{E} \alt \pair{E}{E} \alt \unitValue \\
                           &        & \alt & \begin{caseof}{E} \case{\inl{x}}{E} \case{\inr{x}}{E} \end{caseof} \\
\text{Environments:}       & \Gamma & ::= & {v_1 : \tau^1}_{\beta_1}, ..., {v_n : \tau^n}_{\beta_n} \\
\text{Judgements:}         & J      & ::= & \sequent{\Gamma}{\alpha{}}{\beta}{E:\tau} \\
\end{array} \]

There is a special expression \(\Type\) which is syntactically in
\(\tau\). It is later used when introducing dependent types, since
\(\Type\) lives both as a type and as a term. Translation between proofs
and types is later described in table \ref{encoding-types}.

\begin{table}
The addition, multiplication, and exponentiation can be defined on bounds using $\bSucc{}$
and $\iter{x}{\rho_1}{\rho_2}{\rho_2}$:

$$ \begin{array}{lcll}
\hyper{a}{b}{\bOne}                 & \equiv & \iter{x}{\bSucc{x}}{a}{b}           \\
\hyper{a}{b}{\bSucc{\bOne}}         & \equiv & \iter{y}{\iter{x}{\bSucc{x}}{y}{b}}{\bOne}{\bOne} \\
                                    &   & \hspace{40pt}\textit{Argument $y$ is ignored in this special case.}\\
\hyper{a}{b}{\bSucc{\bSucc{n}}}             & = & \iter{x}{\hyper{x}{a}{n}}{b}{n}     \\
\mathop{h}                          & = & λ g. λ a. λ b. \begin{caseof}{b} \case{\bOne}{a} \case{\bSucc{c}}{\iter{x}{f x a}{c}{a}} \end{caseof} \\
\hyper{a}{b}{n}                     & = & \iter{f}{h(f)}{n}{λ g. g} \\
a + b                               & = & \iter{x}{\bSucc{x}}{a}{b}                \\
a * \bOne                           & = & a \\
a * \bSucc{b}                       & = & \iter{x}{x+a}{b}{a} \\
a ^{1}                              & = & a                   \\
a ^{\bSucc{b}}                      & = & \iter{x}{x*a}{b}{a} \\
a [\bSucc{n}] b                     & = & \begin{caseof}{b} \case{\bOne}{a} \case{\bSucc{c}}{\iter{x}{x[n]b}{c}{a}} \end{caseof} \\
a + b                               & \equiv & \hyper{a}{b}{\bOne}                 \\
a * b                               & \equiv & \hyper{a}{b}{\bSucc{\bOne}}         \\
a ^ b                               & \equiv & \hyper{a}{b}{\bSucc{\bSucc{\bOne}}} \\
\bSucc{\bSucc{\bSucc{\ack{m}{\bSucc{\bSucc{\bSucc{n}}}}}}}  & \equiv & 2[m](\bSucc{\bSucc{\bSucc{n}}}) \\
\end{array} $$

Here $=$ is for definition, and $\equiv$ states equivalence of expressions. To avoid definition oof predecessor, we use equivalence
to express Ackermann function.

$\hyper{a}{b}{n}$, and $a[n]c$ are alternative ways of introducing hyperoperations.

We use hyperoperations for clarity, showing that we can indeed express Ackermann function as bounded iteration of function compositions.

\vspace{10pt}

\caption[Encoding arithmetic]{Encoding arithmetic operations, hyperoperation, and Ackermann function}
\label{encoding-arithmetic}

\end{table}

Notation \(A_{v} \rightarrow^{\alpha{(v)}}_{\beta{(v)}} B\) binds proof
variable \(x\) with type of \(A\) and size variable \(v\), and then
bound in bounds \(\alpha(v)\) for complexity and \(\beta(v)\) for
\emph{depth} of the \emph{normalized} term. We use notation
\(\alpha(v)\) instead of \(\alpha\) to emphasize that both \(\alpha(v)\)
and \(\beta(v)\) are functions of size variable \(v\).

We could attach a pair of bounds to each proposition and judgement
\(A_{(\alpha{}, \beta{})}\) that would describe both complexity
\(\alpha\) of computing the proof and a maximum depth \(\beta{}\) of the
resulting (normalized) witness. However, in most cases, one of these
would be \(1\) or could be inferred from the remaining information.

The \(\occurs{x}{E}\) is a count of free occurrences of variable \(x\)
in term \(E\). Free variables of \(E\) are computed by \(\free(E)\).

\subsection{Inference rules}\label{inference-rules}

With any term variable \(x\) we need to introduce an associated bound
variable \(v\).

\[
\Infer[var]{\sequent{\Gamma}{?}{?}{\isType{A}} & x \in X & v \in V}{\sequent{\Gamma, x_{v} : A}{1}{v}{x : A}}
\]

Sometimes we might want to overestimate proof complexity for the sake of
simplicity: \[
\Infer[subsume]{\sequent{\Gamma}{\alpha_1}{\beta_1}{{e : A}} & \quad \alpha_1 \le \alpha_2 & \quad \beta_1 \le \beta_2}{\sequent{\Gamma}{\alpha_2}{\beta_2}{{\textit{subsume}(e, \alpha_2, \beta_2)}: A}}
\]

Note that subsumption is necessary for \(\mathop{case}\)-expressions.
Below we have typical rules for construction and destruction of basic
types:

\[
\Infer[unit]{}{\sequent{\Gamma}{1}{\beta}{\unitValue : \unitType}}
\]

\[
\Infer[inl]{\sequent{\Gamma}{\alpha}{\beta}{e : {A}}}{\sequent{\Gamma}{\alpha+1}{\beta+1}{\inl{e} : {A \lor B}}}
\rulesep
\Infer[inr]{\sequent{\Gamma}{\alpha}{\beta}{e : {B}}}{\sequent{\Gamma}{\alpha+1}{\beta+1}{\inr{e} : {A \lor B}}}
\]

\[
\Infer[case]{\sequent{\Gamma}{\alpha_{\lor}}{\beta_{\lor}+1}{a : {L \lor R}}
            & \quad \sequent{\Gamma, x_{\beta_{\lor}} : {L}}{\alpha_l}{\beta_l}{l : B}
            & \quad \sequent{\Gamma, y_{\beta_{\lor}} : {R}}{\alpha_r}{\beta_r}{r : B}}
            {\sequent{\Gamma}{\alpha_{\lor}+\max{\alpha_l}{\alpha_r}+1}{\max{\beta_l}{\beta_r}}{\begin{caseof}{a}\case{\inl{x}}{l} \case{\inr{y}}{r} \end{caseof} : {B}}}
\]

\[
\Infer[pair]{\sequent{\Gamma}{\alpha_a}{\beta_a}{{a} : {A}} & \sequent{\Gamma}{\alpha_b}{\beta_b}{{b} : B}}{\sequent{\Gamma}{\alpha_a+\alpha_b+1}{\max{\beta_a}{\beta_b}+1}{(a, b)} : {A \land B}}
\]

\[
\Infer[prl]{\sequent{\Gamma}{\alpha}{\beta+1}{e : {A \land B}} & i \in \{l, r\}}{\sequent{\Gamma}{\alpha+1}{\beta}{\prl{e }: {A}} }
\rulesep
\Infer[prr]{\sequent{\Gamma}{\alpha}{\beta+1}{e : {A \land B}} & i \in \{l, r\}}{\sequent{\Gamma}{\alpha+1}{\beta}{\prr{e }: {B}} }
\]

Please note that notation \(A_v \rightarrow^{\alpha}_{\beta} B\) has a
size variable \(v\) declared as a depth of \emph{normal form proof term}
having type \(A\), and then bounds \(\alpha\) and \(\beta\) apply to the
computation of the result.

\[
\Infer[abs]{\sequent{\Gamma, x_v : A}{\alpha}{\beta}{e : {B}} & x, v \not\in Γ}{\sequent{\Gamma}{\alpha\subst{v}{1}+1}{\beta\subst{v}{1}+1}{\lambda x. e : {A_v {\rightarrow}^{\alpha}_{\beta} B}}}
\]

Note that abstraction increases term depth by one, and application
decreases it by one\footnote{This allows us to correctly treat Church
  encoding.}. All introduction rules (\textit{abs}, \textit{pair},
\textit{inl}, \textit{inr}) increase \(\beta\) by at least
one\footnote{For \textit{unit}, the inner proof term would have null
  depth, since there is no term there. Thus depth is \(0+1\) instead of
  \(\beta+1\).}. Likewise all \emph{non-functional} (data) elimination
rules (\textit{case}, \textit{prl}, \textit{prr}) decrease depth
expected from the resulting normal form \(\beta\) by one.

\[
\Infer[app]{\sequent{\Gamma}{\alpha_1}{\beta_1}{e : {A_v {\rightarrow}^{\alpha_2}_{\beta_2} B}}
           &\sequent{\Gamma}{\alpha_3}{\beta_3}{a : A}}
           {\sequent{\Gamma}{\alpha_1+\alpha_2\subst{v}{\beta_3}+\alpha_3}{\beta_2\subst{v}{\beta_3}}{e \; a : B}{} }
\]

Please note that these rules all maintain bounded depth with no
unbounded recursion. We add an explicit bounded recursive definition
(like the definition of the closure) with this rule:

\[
\scalebox{1.27}{
\Infer[rec]{\sequent{\Gamma}{\alpha_1}{\beta_1}{{f : A_v {\rightarrow}^{\alpha_2}_{\beta_2} A}}
           &\sequent{\Gamma}{\alpha_3}{\beta_3}{k : B}
           &\sequent{\Gamma}{\alpha_4}{\beta_4}{a : A}
            }
           {\sequent{\Gamma}
                    {\alpha_1+\iter{v}{\alpha_2}{\beta_3}{\beta_4}+\alpha_3+\alpha_4}
                    {\beta_1\subst{v}{\iter{v}{\beta_2}{\beta_3}{\beta_4}}}
                    {{\rec{f}{k}{a}} : B}{}}
}
\]

Here the depth of the term must decrease at each step of the recursion.
With the exception of \emph{subsume}, and \emph{rec} these are all
reinterpretations of rules for intuitionistic logic
\autocite{intuitionism,dalen,curryhoward}, enriched with bounds on the
proof length \(\alpha\) and normalized depth of term \(t\) namely
\(\depth{t}\) as depth expression \(\beta\). The rule \emph{rec} allows
for explicitly \emph{bounded recursion}, as opposed to traditional
approaches that rely on an unbounded fixpoint\footnote{Fixpoint may not
  exist, thus leading not only to arbitrarily long computation, but also
  to undecidability in cases where computation may never end.}.

Note that we may quantify on higher order values, but we cannot recurse
indefinitely: there is always a limit to a number of function
compositions allowed. Power of bounded function composition gives an
explicit limit to Peano Arithmetic
induction\autocite{goodstein-unprovable-in-peano}: any computational
application of Peano induction is unbounded. At the same time, we can
use multiple recursions over bounded number of functions, terms, not
just natural numbers. Wired-in explicit bounding also allows us to prove
termination of arbitrary ``towers'' of function compositions, like
hyperoperations
\autocite{bennet-hyperoperations,hilbert-hyperoperations,ackermann-study,ackermann-superpowers}\footnote{See
  table \ref{encoding-arithmetic} to see how hyperoperations and
  Ackermann function can be encoded using \textit{iter} for bounded
  iteration of function composition.}, including Goodstein functions
that cannot be proven \emph{within} PA itself\autocite{goodstein}.

\subsubsection{Implicit universal
quantification}\label{implicit-universal-quantification}

In the propositional logic above, provability allows us to confirm
statements with \(\forall\) for all variables on top. Given that
statement of existence of bounded proof term \(x\) for witness bounded
by result size \(v\) can be interpreted in the following way in unbound
logic \(∃v∈\N^{+}. ∃ x. \depth{x}≤v\) .

So \(a_v \rightarrow^{\alpha}_{\beta} b\) becomes the following
statement in unbounded logics
\(∀v.∀a.(\depth{a}<v) \rightarrow a \rightarrow b \land \depth{b} < β ∧ c(b) < α\).
That is: we can infer that fact for all \(a\) below an arbitrarily large
depth, and bound the depth and computational complexity of the resulting
witness.

This concludes our treatment of Ultrafinitist Propositional Logic
(UFPL).

\subsubsection{Quantification with dependent
types}\label{quantification-with-dependent-types}

It is customary in constructive mathematics and theorem proving to use
dependent types instead of usual universal and existential quantifiers
\autocite{dependent-types}.

Please note that just like one can define intuitionistic propositional
logic~with just implication and then encode both sum and product
types\autocite{church-encoding}, so \(\Pi\) type can express both
universal quantification and plain implication, while \(\Sigma\) type
can express both existential quantification and product type. Since
implication can already express sum and product types in polymorphic
calculus, we will only show how to modify rules for implication and
lambda to make the \(\Pi\) type that corresponds to universal
quantification.

While it is usual to introduce universal quantification directly in
calculi without proof terms, we will introduce them with \(\Pi\) types,
like is now customary in dependently typed languages.

First we need a rule to introduce a type variable:

\[
\Infer[tyvar]{\sequent{Γ}{\alpha}{\beta}{\isType{t}}}{\sequent{\Gamma, v < β, x_v : t}{1}{v}{x : t}}
\]

This rule allows us to use variables at type level, and together with
\(\Pi\) and \(\Sigma\) types allow to express quantification.

For the inequalities, it suffices to ensure that they are not cyclic and
thus unsatisfiable. Note that inequality stems from the fact that value
is always no longer than its encoding as a type.

\[
\Infer[forall-form]{\sequent{Γ}{\alpha_1}{\beta_1}{\isType{A}} & \sequent{Γ, x_v : A}{\alpha_2(v)}{\beta_2(v)}{\isType{B}} & x, v \not∈Γ }{\sequent{\Gamma}{\alpha_1+\alpha_2\subst{v}{1}}{\max{\beta_1}{\beta_2\subst{v}{1}}+1}{\isType{\depprod{x_v}{A}{\alpha_2}{\beta_2}{B}}}}
\]

Please note that similarly to the treatment of lambda abstraction as
proof of implication, we estimate the computational cost of dependent
product by substituting free variables with 1, but now we still need to
consider the same substitution in the resulting \(\Pi\) type.

Treatment of universal quantifier bears usual
similarity\autocite{dependent-types} to \emph{abs} and \emph{app} rules:

\[
\Infer[forall-intro]{\sequent{\Gamma, x_v : A}{\alpha}{\beta}{e : {B}}}{\sequent{\Gamma}{\alpha\subst{v}{1}+1}{\beta\subst{v}{1}+1}{\lambda(x : A). e : \depprod{x_v}{A}{\alpha}{\beta}{B}}}
\]

Elimination works the same way as for usual application, since
computation works after type erasure.

\[
\Infer[forall-app]{\sequent{\Gamma}{\alpha_1}{\beta_1}{e : {A_v {\rightarrow}^{\alpha_2}_{\beta_2} B}}
           &\sequent{\Gamma}{\alpha_3}{\beta_3}{a : A}}
           {\sequent{\Gamma}{\alpha_1+\alpha_2\subst{v}{\beta_3}+\alpha_3}{\beta_2\subst{v}{\beta_3}}{e \; a : B\subst{x}{A}}{} }
\]

This allows us to replace implication and by extension, all UFPL~types.
It also allows for quantification of higher order values.

We leave introducing \(\Sigma\)-types to interested students of type
theory, since they are not essential to our argument that we may have a
decidable higher order logic.

Since introduction creates the proof term in the same way, the proof
terms can be enumerated in the same way as shown in section
\ref{decidability} on page \pageref{decidability}.

\section{Application of the logic}\label{application-of-the-logic}

\subsection{Using proofs}\label{using-proofs}

Each proof ultimately leads to a judgment
\(\sequent{\Gamma}{\alpha}{\beta}{e : A}\). We may resolve all upper
bound variables \(v_1, v_2, ..., v_n\) in the \(\alpha\) to get an upper
bound on computational complexity of the statement, and in \(\beta\) to
get an upper bound on normalized term resulting from the proof. This way
all proofs are ultra-finitary statements: Only as long as \(\alpha\) is
less than our assumed limit, we will consider the proof valid and proof
computation to be available within the given time.

\subsection{Simplifying upper bounds}\label{simplifying-upper-bounds}

Our inference rules rely on computing upper bounds and their inequality.
Here we note a few inequalities that simplify reasoning about these
bounds, albeit at the cost of making them somewhat looser.

First, we note that all variables are positive naturals because they
represent the data of non-zero size: \(x \geq 1\).

That means that the following laws are true, assuming that
\(x, y, ... \geq 1\) are data size variables in the environment,
\(1 \leq e \leq f\) and \(1 \leq g \leq h\) are arbitrary positive
expressions, and \(a, b, c...\geq 1\) are constants. For easier use, the
rules are presented in left-to-right order, just like conventional
rewrite rules.

\begin{table}

\[\begin{array}{lrcl}
(1) & a*x^e + b*x^f          & \leq & (a+b)*x^f      \\
(2) & a*x^e*y^g              & \leq & a*x^f*y^h      \\
(3) & \iter{v}{e}{g}{x}      & \leq & \iter{w}{f}{h}{x} \\
(4) & \iter{v}{v*a}{e}{x}    & =    & a^e*x          \\
(5) & \iter{v}{v+a}{e}{x}    & =    & x+a*e          \\
(6) & \iter{v}{v^e}{g}{x}    & =    & x^{e^g}        \\
\end{array}\]

Assumptions:

\begin{itemize}
\tightlist
\item
  \(x, y \geq 1\) are data size variables in the environment,
\item
  \(1 \leq e \leq f\) and \(1 \leq g \leq h\) are arbitrary positive and
  increasing expressions,
\item
  \(a, b, c...\geq 1\) are constants.
\end{itemize}

\vspace{8pt}

\caption{Simplication of bounds. May rewrite left to right.}
\label{simplification-of-bounds}

\end{table}

We may thus use these rules to loosen the bound in such a way as to
reduce the size of the bound expression and make it a sum of a single
term in all variables and an additional constant term. This reduction
may be delayed until we have bound to verify.

We may use inference rules leaving ``type holes''\autocite{type-holes}
instead of bounds, which could be named ``bound holes'', and let them be
filled by the framework interpreter.

\subsection{Reduction}\label{reduction}

Reduction relation is defined as small step semantics
\autocite{structural-operational-semantics} in order to preserve number
of computational steps made over the course of evaluation. See table
\ref{table:reduction} on page \pageref{table:reduction}.

\begin{table}[!t]

$$
\Infer[eval-case-arg]{\eval{k}{e}{e'}}{\eval{k}{\begin{caseof}{e}\case{\inl{x}}{b} \case{\inr{y}}{c} \end{caseof}}{\begin{caseof}{e'}\case{\inl{x}}{b} \case{\inr{y}}{c} \end{caseof}}}
$$

$$
\Infer[eval-case-left]{k=\occurs{x}{b}}{\eval{k}{\begin{caseof}{\inl{a}}\case{\inl{x}}{b} \case{\inr{y}}{c} \end{caseof}}{b \subst{x}{a}}}
$$

$$
\Infer[eval-case-right]{k=\occurs{y}{c}}{\eval{k}{\begin{caseof}{\inr{a}}\case{\inl{x}}{b} \case{\inr{y}}{c} \end{caseof}}{c \subst{y}{a}}}
$$

$$
\Infer[eval-app]{k=\occurs{x}{e}}{\eval{k}{(\lambda x. e) f}{e \subst{x}{f}}}
$$

$$
\Infer[eval-sum]{\eval{n}{e}{e'} & i \in \{l, r\}}{\eval{n}{\mathop{in_i}(e)}{\mathop{in_i}(e')}}
$$

$$
\Infer[eval-prl]{}{\eval{1}{\prl{\pair{a}{b}}}{a}}
\rulesep
\Infer[eval-prr]{}{\eval{1}{\prr{\pair{a}{b}}}{b}}
$$

$$
\Infer[eval-pairleft]{\eval{n}{a_l}{a'_l}}{\eval{n}{\pair{a_l}{a_r}}{\pair{a'_l}{a_r}}}
$$

$$
\Infer[eval-pairright]{\eval{n}{a_r}{a'_r}}{\eval{n}{\pair{a_l}{a_r}}{\pair{a_l}{a'_r}}}
$$

\caption{Reduction rules}

\label{table:reduction}

\end{table}

When performing application, we expect substitution to take work
proportional to the number of occurrences of the variable, like changing
links on directed acyclic graph of the term.

For discussion of efficient reduction of lambda terms please read
\autocite{optimal-reduction,optimality-linear-reduction}, since here we
focus on demonstration with a simplified cost model.

\subsection{Self-encoding}\label{self-encoding}

\subsubsection{Natural numbers}\label{natural-numbers}

\begin{table}[!t]

\[
\begin{array}{rcllll}
\Nat_{\beta} & = & \splice{\rec{x}{\interp{\unitType \lor x}}{\beta}{\interp{\unitType}}} \\
\Zero   & = & \inl{\unitValue} & :^{1}_{1} & \Nat_1 \\
\aSucc  & = & \lambda x_v \rightarrow^{1}_{v+1} \inr{x} & :^{1}_{v+1} & \Nat_v \rightarrow Nat_{v+1} \\
\end{array}
\]

\caption{Encoding natural numbers}
\label{encoding-naturals}

\end{table}

In this section we will encode bounds, propositions (types) and proof
terms as proof terms within UFPL. Thus \(\interp{..}\) corresponds to
LISP quote.

Below we use notation \(\dbi{v}\) for de Brujin index of the variable
\autocite{de-brujin-index}.

\subsubsection{Encoding bounds}\label{encoding-bounds}

Now we may encode bounds (table \ref{encoding-bounds}), types (table
\ref{encoding-types}), and proof terms (table \ref{encoding-terms}).

\begin{table}

\[ \begin{array}{lcl}
\Var_{\beta}                       & = & \Nat_{\beta} \\
\Bound_{\beta+1}                   & = & \Var \lor \Nat_{\beta} \lor \unitType \lor \pair{\Bound_\beta}{\Bound_\beta} \\
                                   & \lor & \pair{\Bound_\beta}{\Bound_\beta} \\
                                   & \lor & \pair{\Bound_\beta}{\Bound_\beta} \\
                                   & \lor & \pair{\Bound_\beta}{\pair{\Bound_\beta}{\Var}} \\
                                   & \lor & \pair{\Bound_\beta}{\pair{\Var}{\Bound_{\beta}}}\\
\interp{v}                         & = & \inl{\inl{\inl{\dbi{v}}}} \\
\interp{i}                         & = & \inl{\inl{\inr{i}}} \\
\interp{\unitValue}                & = & \inl{\inl{\inr{\unitValue}}} \\
\interp{\rho_1 + \rho_2}           & = & \inl{\inr{\inr{(\interp{\rho_1}, \interp{\rho_2})}}} \\
\interp{\rho_1 * \rho_2}           & = & \inr{\inl{\inl{(\interp{\rho_1}, \interp{\rho_2})}}} \\
\interp{\rho_1^{\rho_2}}           & = & \inr{\inl{\inr{(\interp{\rho_1}, \interp{\rho_2})}}} \\
\interp{\iter{v}{\rho_1}{\rho_2}{\rho_3}}   & = & \inr{\inr{\inl{(\dbi{v},\interp{\rho_1}, (\interp{\rho_2}, \interp{\rho_3}))}}} \\
\interp{\rho_1 \; \subst{v}{\rho}} & = & \inr{\inr{\inr{(\interp{\rho_1},(\interp{\rho_2}, \dbi{v}))}}} \\
\end{array}
\]

\caption{Encoding bounds}
\label{encoding-bounds}

\end{table}

\begin{table}

\[
\begin{array}{lcl}
\interp{A \lor  B}           & = & \inl{\inl{\pair{\interp{A}}{\interp{B}}}} \\
\interp{A \land B}           & = & \inl{\inr{\pair{\interp{A}}{\interp{B}}}} \\
\interp{A_v \rightarrow^{\alpha}_{\beta}  B} & = &
\inr{\inl{(\lambda x : A . \interp{B}, (\lambda v:\Nat_v. \interp{\alpha}, \lambda v:\Nat_v. \interp{\beta}))}}          \\
\interp{\unitType}           & = & \inr{\inr{\unitValue}}          \\
\end{array}
\]

\caption{Encoding types}
\label{encoding-types}

\end{table}

\begin{table}

\[
\begin{array}{lcl}
\interp{x_v}                   & = & \inl{\inl{\inl{\inl{\pair{{\dbi{x}}}{v}}}}} \\
\interp{\mathop{subsume}(A,B)} & = & \inl{\inl{\inl{\inr{\pair{\interp{B}_{\Bound}}{\interp{A}}}}}} \\
\interp{\mathop{unit}}         & = & \inl{\inl{\inr{\inl{\unitValue}}}} \\
\interp{\inl{A}}               & = & \inl{\inl{\inr{\inr{{A}}}}} \\
\interp{\inr{A}}               & = & \inl{\inr{\inl{\inl{{A}}}}} \\
\interp{\prl{A}}               & = & \inl{\inr{\inl{\inr{{A}}}}} \\
\interp{\prr{A}}               & = & \inl{\inr{\inr{\inl{{A}}}}} \\
\interp{\pair{A}{B}}           & = & \inl{\inr{\inr{\inr{\pair{\interp{A}}{\interp{B}}}}}} \\
\interp{A B}                   & = & \inr{\inl{\inl{\inl{\pair{\interp{A}}{\interp{B}}}}}} \\
\interp{\lambda x_v. A}        & = & \inr{\inl{\inl{\inr{\pair{\pair{\dbi{x}}{\dbi{v}}}{\interp{A}}}}}} \\
\interp{\rec{v}{A}{B}{C}}      & = & \inr{\inl{\inr{\inl{\pair{\pair{\dbi{v}}{\interp{A}}}{\pair{\interp{B}}{\interp{C}}}}}}} \\
\end{array}
\]

\caption{Encoding terms}
\label{encoding-terms}

\end{table}

This encoding allows us to make operations on types akin to generic
programming in Haskell \autocite{generic}.

Our inference rules rely on computing bounds and their inequalities.
Given that all variables are positive naturals because they represent
the data of non-zero size: \(x \geq 1\), we may simplify these bounds
with a set of simple inequalities.

\subsubsection{Encoding proof terms}\label{encoding-proof-terms}

Note that every type term in \emph{normal form} is longer than its own
type.

\begin{theorem}[Encoding]

All bound, type, proof, or proposition of UFPL can be encoded as a proof
term of UFPL.\end{theorem}

Details are visible in the tables
\ref{encoding-naturals}-\ref{encoding-terms}.

\section{Properties of the logic}\label{properties}

When implementing the computation seems straightforward, we will just
establish the finite limit for the computation that should be taken as a
proof. That is what we describe as \emph{problem is decidable by the
limit of} a given complexity. This approach explicitly describes
undecidable problems as those that require an infinite number of steps
to solve.

\subsection{Consistency}\label{consistency}

Here we will only use well-known proof of consistency of intuitionistic
logic \autocite{intuitionism,dalen,curryhoward}\footnote{The proof above
  is totally independent of previous conjectures.}. We do not use the
self-encoding presented in section \ref{self-encoding}.

\begin{theorem}[Consistency of UFL]

UFPL is consistent, if intuitionistic propositional logic is
consistent.\end{theorem}

\begin{proof}

After elision of
bounds\footnote{Elision of bounds is only used once to prove consistency.},
we interpret the rule \emph{subsume} as \(\mathop{id}=\lambda x.x\).
Then we see the standard proof rules for intuitionistic logic. The
consistency follows from the consistency of intuitionistic
logic.\end{proof}

\subsection{Expressivity}\label{expressivity-prf}

\begin{theorem}[At least as expressive as PRA.]

UFL can express all Primitive Recursive programs.\end{theorem}

\begin{proof}

It is easy to show that our logic can emulate bounded loop
programs\autocite{loop-language} which has power equivalent to primitive
recursive functions\autocite{pra}. Every bounded \emph{loop} can be
encoded by \(\iter{v}{\textit{loop}}{x}{n}\), then every flat logical
statement can be encoded with a tuple containing states of the
variables.\end{proof}

One could muse that this class does not cover all Bounded Turing
Machine\autocite{time-bounded-turing} programs. In order to support
these, we would need to define more general bounding functions.

One can replace upper bound expressions with arbitrary bounding
functions expressed in simply typed lambda calculus (see section
\ref{self-encoding}). These are the operations used in inference rules.
However, such functions are more difficult to bound and compute
themselves.

It has been proven that any function whose complexity is \emph{bounded}
by primitive recursive function is also primitive
recursive\autocite{pra-bounding}, which means that estimating our
complexities could become an impossibly long endeavour, but logically
consistent one.

To give an example of simplified Ackermann function which is the best
known example of function beyond PRA \autocite{sudan,ackermann},
evaluation takes
\(\mathop{A}(5) = 2^{2^{2^{2^{16}}}} - 3\)\autocite{oeis-simplified-ackermann}.
That means that these evaluations quickly get out of hand and indeed
outside of any reasonable limits.

The encoding of Ackermann function is through hyperoperation in table
\ref{encoding-arithmetic}.

\subsection{Bounded Turing completeness}\label{turing-completeness}

An evidence of stronger expressivity may be found by encoding bounded
Turing Machine programs in UFPL. This proof uses encoding similar to
\ref{self-encoding}, but for a Turing Machine. For a reference on
encoding of Turing machines in lambda calculus see
\autocite{encoding-turing-machines}.

For any complexity bound \(f(x)\) expressible in the language of
ultrafinitist logic, and an algorithm that satisfies it and emulation
function with complexity of \(e(f)(x)\) -- that is an encoding \(e(f)\)
of \(f\), applied to the argument \(x\) of \(f\) -- which we can encode
this emulation as a bound.

\begin{theorem}[Emulation complexity] \label{emulation-complexity}

Assume a time complexity \(c(x)\) for program (or proof) \(s\) that can
be encoded as UFL bounds. If we can emulate (encode evaluation) of
\(f(x)\) with an overhead \(e\) for each step, then we can prove that
complexity of evaluating \(s\) is \(e*c(x)+cc(x)\).

Where \(cc(x)\) is complexity of evaluating complexity bounds for the
encoding \(c(x)\).\end{theorem}

\begin{proof}

Given each step of emulation encoded as \(s(x)\), where \(x\) is a
current state, emulation with a complexity function encoded as \(f(x)\)
can be executed by \(\iter{v}{s}{f(x)}{x}\).

Assuming that \(e(f)\) is function emulation in UFPL, we can write proof
expression \(\iter{x}{e(f)(x)}{e(c)(x)}{x}\). This expression evaluated
encoded \(s\) and has exactly the assumed complexity\end{proof}

The most complex part of the proof may be logically inferring the right
complexity \(c(x)\) and totality of the function \(f\) within this
number of steps.

\begin{theorem}[Bounded Turing Machine emulation]

For programs of Bounded Turing Machine \(f\) over alphabet size \(|a|\)
and number of states \(|s|\) with complexity that can be encoded in UFL,
we can prove time complexity of \(\lg_2(|a|)+\lg_2(|s|)*|c|+|cc|\) with
UFPL. Note that \(|cc|\) is cost to evaluate complexity function
itself.\end{theorem}

\begin{proof}

We use emulation argument for Bounded Turing machines that may be
limited by bounds described above, it is too with
\(\bigoh{\log^2(a)*f(x)}\), where \(a\) is the bound on the size of
alphabet and number of states of the machine.

For the Bounded Turing machine we encode tape as pair of lists, with
current position at the top of both lists.

Then we encode the following steps:

\begin{itemize}
\tightlist
\item
  examine the alphabet character: \(\bigoh{\log(|a|)}\)
\item
  examine finite state machine for a character: \(\bigoh{\log(|s|)}\)
\item
  move one step up or down the tape by moving the top from one line list
  to the other: \(\bigoh{1}\);
\item
  if we want to write at the current position, we take the top element
  from the right list, and put the new one.
\end{itemize}

Together they make a single step of the Turing machine at the cost of
\(\bigoh{\log(|a|) + \log(|s|)}\).\end{proof}

We encode variable bindings as a dictionary with cost of
\(\bigoh{\lg_2{|\Var|}}\), where \(\Var\) is number of variables used.
All operations not involving substitution should remain at \(\bigoh{1}\)
complexity within emulation.

\begin{lemma}[Self-emulation]

ULF self-emulation of function with integral bound \(|cc|\) is feasible
within \(\bigoh{|cc|*\lg_2{|\Var|}}\).\end{lemma}

Overall we can infer that for each algorithm of bounded complexity \(B\)
that we may encode in ULF, we may use ULF self-emulation to find a proof
with complexity of at most \(e*B*\lg_2{|\Var|}\). All steps of ULF are
\(\bigoh{1}\) with respect to inferred bounds on computational
complexity, with the exception of function application and variable
substitution which are \(\bigoh{\lg_2{|\Var|)}}\)

\begin{theorem}[Emulation completeness]

If the bounds that can be encoded within the bounds function, the UFPL
is complete for proving its own bounds up to the cost of self-emulation
\(e\).\end{theorem}

Since we can encode any statement in UFPL in UFPL itself, this likely
would mean the proof of emulation completeness can be written in the
UFPL itself.

There are complex ways of proving completeness that apply in the realm
of non-idempotent intersection type systems, but they use a more
abstract notion of complexity\autocite{tight-bounds}.

\subsection{Decidability of bounded statements}\label{decidability}

\begin{theorem}[Decidability]

Every valid proposition with a \emph{fixed bound on input} \(n\) can be
checked by enumerating inputs, and is thus decidable.\end{theorem}

This comes at the cost of complexity that increases by
\(\alpha(n)*a^n\), where \(n\) is the depth of input, since we need to
enumerate all inputs of depth \(n\). Proof follows directly from
enumeration, and bounds.

\begin{proof}

Let's try to enumerate the terms that can be constructed for a given
bounds, without using \emph{subsume} rule:

Bounds function will contain a given number \(n\) of successor
functions, and \(\mathop{iter}\) expressions. Each time we make a single
inference rule, and construct a slightly more complex term, we add a
successor function or \(\mathop{iter}\) expression, the number of
different proof tree shapes equals to the number of ternary trees
constructible with \(n\) nodes. This number is defined as OEIS A001764
\autocite{oeis-ternary-trees} and given by the algebraic term
\(\frac{\binom{3n}{n}}{2n+1}\). Given that for each of \(n\) nodes we
can choose one of the 12 rules (\(12^n\) different choices), we have at
most \(\bigoh{\frac{12^n\dbinom{3n}{n}}{2n+1}}\) different proofs with
the complexity given by a term of \(n\) nodes.

With the inclusion of subsumption rule, we can only decrease the \(n\),
so at most
\(\sum\limits_{i=1..n}\bigoh{\frac{12^n}{2n+1}\dbinom{3n}{n}}=\bigoh{12^n\dbinom{3n}{n}}\).

Since number of proofs is finite, we can decide the provability after
they are exhausted.\end{proof}

Thus all judgements with bounds are decidable. This property is shared
with some other \emph{resource bounded logics} \autocite{tight-bounds}.

Given that exponential lower bound has been established for
implicational intuitionistic logic \autocite{exponential-lower-bound},
we expect that lower bound for ultrafinitist logic will also be
exponential and thus the proposed bound is asymptotically tight.

\subsection{Paradox of undecidability}\label{paradox-of-undecidability}

Expressing any statements about undecidability implicitly requires
unbounded computational effort. Since all our proofs and arguments are
explicitly bounded, there is no room to state undecidability. Thus we
conclude that this paradox is removed from ultrafinitist logic:
statement of undecidability is \emph{invalid} as a proposition. All
\emph{valid propositions} are decidable.

This is not as outrageous as it superficially seems, since we already
know that computation models that would allow transfinite number of
steps would also make all functions computable \autocite{hamkins}.

\subsection{Finitary completeness}\label{finitary-completeness}

Let's assume we have upper bounds on all variables within an
intuitionistic theorem.

Can we prove it with UFL?

\begin{theorem}[Preservation of bounded intuitionistic theorems]
\label{preservation-of-bounds}

Any intuitionistic theorem bounded by definite integers in UFL can be
proven in UFL.\end{theorem}

\begin{proof}

Let's enumerate complexities of computing intuitionistic proof for a
given set of inputs bounded by given value. We may enumerate these
proofs, and thus take maximum length of the computation. This maximum
length will be upper bound on all proofs.\end{proof}

This proof uses \ref{decidability}, and \ref{self-encoding}. Naturally
this means that all statements with bounds but proof without bounds will
also have proof with bounds.

\section{Related work}\label{related-work}

The philosophical problem with transfinite arguments has been spotted
long before
\autocite{explicit-finitism,podnieks-finitism,yessenin-criticism,gefter,lenchner}.
Automatic theorem provers like Coq require a monotonically decreasing
bounding function in the ordinal domain for each inductive definition
\autocite{well-founded-recursion,coq-recursion,coq-partial-recursion}.
This makes all recursive definitions
\emph{well-founded}\autocite{well-founded-recursion}, but since
transfinite ordinals are permitted, it also allows theories outside
computable universe.

The computation of a bounding function may turn out to take unfeasible
amount of time. Cost calculi for functional languages attempt to assign
cost to certain operations in order to reason about time and space
complexity \autocite{cost-calculi}. But these approaches do not require
all proofs and propositions to carry the cost as we do.

Philosophers have postulated distinction between feasible computations
and unfeasible ones \autocite{yessenin-criticism}, however it was
considered unclear whether it is possible to realize this distinction on
the basis of a logic \autocite{constructivism-review}, with some
claiming that such a logic could not be consistent
\autocite{wangs-paradox,strict-finitism-refuted-q}.

There exist logics that implicitly constrain computational complexity of
the proofs, for example Bounded Arithmetic
\autocite{buss,bounded-arithmetic,implicit-complexity} that is
restricted to computations in polynomial time. However, most of them are
significantly weaker than class of primitive recursive functions, which
is widely considered to contain most useful programs. This would put the
logician in a position of trying to state a widely known facts about
objects that are inexpressible within the logic.

\section{Discussion}\label{discussion}

\subsection{Explicit bounds versus implicit structural
recursion}\label{explicit-bounds-versus-implicit-structural-recursion}

It is long known that unbounded logics may give rise to paradoxes
\autocite{girard-paradox,analysis-girard-paradox}, and the use of
implicit techniques \autocite{implicit-complexity}, including bounded
recursion \autocite{bounded-arithmetic}, structural recursion
\autocite{structural-recursion}, well-founded sets
\autocite{well-founded-recursion}, or predicative bounding
\autocite{Feferman-predicativism} were developed.

Using explicit bounds provides a more obvious solution, which is easier
to prove correct, and parallels development of an explicit mathematical
limit \autocite{epsilontics,delta-epsilon}, starting from Eudoxus'
method of exhaustion \autocite{eudoxus}, through implicit notion of
terminus \autocite{terminus} to a modern concept of a mathematical limit
of a function \autocite{weierstrass}\footnote{Interestingly
  delta-epsilon definition is formalizable for computable functions, by
  assuming that as \(n\) approaches the limit the smaller computable
  environment is taken. Of course both delta and epsilon would have to
  be a finite expansions (approximations) instead of possibly
  transcendental value. This would give a definition of ``computable
  limit''.}.

\subsection{Open problem of directly proving bounded Turing
completeness}\label{open-problem-of-directly-proving-bounded-turing-completeness}

Note that the proof above mentions Turing completeness, if we can prove
that all bounds can be expressed by the bounds functions defined above.
While the usual examples of fast growing functions like Ackermann
\autocite{sudan,ackermann,oeis-simplified-ackermann}, or Goodstein
\autocite{goodstein,goodstein-unprovable-in-peano} are expressible by
bounded composition of functions, the clarity is still elusive. (Of
course such fast growing functions would quickly surpass any reasonable
limit.)

We still search for a proof of bounded Turing completeness that would
not use a recursive argument, where we replace bound functions with
arbitrary bounded lambda expressions. That is because our induction
principle would have to be more complex to include the latter.

Interested student may prove Turing completeness by encoding to Kleene
normal form with iteration on top \autocite{hermes}. In this logic,
Kleene normal form may be made explicit.

\subsection{Computability as foundation of
mathematics}\label{computable-foundations}

Finite descriptions of the proofs and their objects are most rational
foundations of mathematics. These objects are all definable by bounded
Turing computability.

Attempts to define \emph{hypercomputation} beyond bounded Turing machine
immediately lead to physical impossibilities
\autocite{intractable-hypercomputation}. At the same time computability
or bounded Turing machine and total computable functions have been
translated between multiple mathematical models. Hence we conjecture
that the only mathematical proof principle that is immune to rational
doubt is the bounded Turing machine, and ultrafinitist logic.

A logic that allows expression of any bounded Turing function and
nothing else could be rightly called a logic of computable functions,
and a best candidate for encoding foundations of mathematics.
Alternative attempts to narrow set theory by predicativism
\autocite{Feferman-predicativism,Schuette-predicativism} are subject to
critique\autocite{criticism-of-predicativism} that motivates further
search.

That is because we can encode axioms that are incomputable as function
parameters with assumed types, and use these to prove or disprove
theorems of traditional axiomatic theories without endangering
consistency of the underlying logical framework.

\subsection{Automated theorem proving}\label{automated-theorem-proving}

Interesting avenue for future work would be to define a full type
theory, dependently typed language and an automatic prover for these
inference rules. Improving on the bound of \(\bigoh{\binom{3n}{n}}\) for
deciding subtheorems would be possible, since we only need to consider
normal forms. It would be exciting to prove metatheoretic results about
the UFL in itself, and verify it with an automatic theorem prover.

Since meta-reasoning always results in longer proofs than original
theorems, the UFL may also allow us to prove consistency of
ultrafinitist arithmetic, enabling to second Hilbert problem
\autocite{hilbert-problems}, and potentially allowing self-verifiable
formalization of mathematics.

Theories for uncomputable are only indirectly formalisable within such
framework as functions taking uncomputable actions (like infinite
recursor of Peano arithmetic) as arguments. Previously created theories
are prone to high complexity and errors due to difficulty at maintaining
expressivity and consistency together. Simplicity of proving the
hierarchy of universes as hierarchy of complexities, and expanding
ultrafinitist logic with strongly normalizable dependent types gives us
hope that such automated theorem proving framework would be simpler.

Since the logic includes upper bounds for all functions, we may use
these and proof irrelevance to automatically and safely optimize proofs
as well. For example, we could automatically replace computation of
naturals defined by successor function with computation defined on
positional binary numbers.

\subsection{Proving decidability in strictly finite
domain}\label{proving-decidability-in-strictly-finite-domain}

The explicit bounding of all objects, including proofs in this work is
used to prevent undecidability within finite domain
\autocite{incompleteness-in-finite-domain}.

\section{Conclusion}\label{conclusion}

We have shown a possible consistent logic for inference with a strictly
bounded number of steps. This allows us to limit our statements by the
length of acceptable proof, and thus define statements that are both
true, and computable within Bremermann-Gorelik limit
\autocite{bremermann} This inference system explicitly bounds both the
length of the resulting proof, and the bounds on the depth of the
normalized result term. This allows avoiding inconsistencies suggested
by philosophical work, and at the same time steers away from relatively
weak logics with implicit complexity like Bounded Arithmetic
\autocite{bounded-arithmetic}, which capture polynomial time hierarchy.
It also shows how much we gain by making explicit bounds, since these
may be tighter than with implicit complexity approaches. While Emulation
Complexity is the powerful approach to proving expressivity of this
logic, it would be nice to see a proof with tighter bounds on what we
may prove with it\footnote{A promising avenue of work would be proving
  that \emph{amortized complexity} by replacing single bound variable by
  a vector of monotonic bound variables. Another approach would be
  attempt to obtain tighter bounds directly by separately counting
  beta-reduction steps and substitutions, instead of all reduction steps
  and substitutions together \autocite{tight-bounds}.}.

We strive to prove that all bounded computable functions are expressible
within this framework, and thus we propose this logic as a ``logic of
practical computability''.

\bibliography{ultrafinitist.bib}

\end{document}